\documentclass[preprint,3p]{elsarticle}

\usepackage{nameref}
\usepackage{hyperref}
\hypersetup{
    colorlinks,
    linkcolor={red!50!black},
    citecolor={blue!50!black},
    urlcolor={blue!80!black}
}

\usepackage{booktabs,amsthm,amsmath,multirow,amssymb}
\usepackage{mathtools}
\usepackage{mathabx}
\usepackage{csquotes}
\usepackage{comment}
\usepackage[normalem]{ulem}

\usepackage{placeins}
\usepackage{float}
\usepackage{todonotes}
\usepackage{complexity}
\usepackage{subcaption}

\newcommand{\cB}{\mathcal{B}}
\newcommand{\cC}{\mathcal{C}}

\newcommand{\cX}{\mathcal{X}}

\definecolor{mygreen}{rgb}{0.0,0.6,0.0}

\newlength{\contentheight}

\newcommand{\ig}{\includegraphics}

\newcommand{\dir}[1]{\overrightarrow{#1}}

\newcommand{\ldir}[1]{\overleftarrow{#1}}

\newcommand{\bs}[1]{\backslash{#1}}

\newcommand{\Planar}{\textsc{Planar}}
\newcommand{\ThreeSAT}{\textsc{3-SAT}}
\newcommand{\PlanarThreeSAT}{\Planar~\ThreeSAT}

\newtheorem{problem}{Problem}

\newtheorem{theorem}{Theorem}[section]
\newtheorem{lemma}[theorem]{Lemma}
\newtheorem{proposition}[theorem]{Proposition}

\newtheorem{corollary}[theorem]{Corollary}

\theoremstyle{remark}

\renewenvironment{proof}{\medskip \par
\noindent \textbf{Proof.}
}{\hfill$\qed$\medskip}

\journal{Discrete Applied Math}

\begin{document}

\begin{frontmatter}
\title{Note about the complexity of the acyclic orientation with parity constraint problem }

\author[1,3]{Sylvain Gravier}
\ead{sylvain.gravier@univ-grenoble-alpes.fr}
\author[1,3]{Matthieu Petiteau \corref{cor1}}
\ead{matthieu.petiteau@univ-grenoble-alpes.fr}
\author[2,3]{Isabelle Sivignon}
\ead{isabelle.sivignon@univ-grenoble-alpes.fr}

\cortext[cor1]{Corresponding author}

\affiliation[1]{organization={Univ. Grenoble Alpes, CNRS, Institut Fourier},
            city={Grenoble},
            postcode={38000},
            country={France}}

\affiliation[2]{organization={Univ. Grenoble Alpes, CNRS, Grenoble INP, GIPSA-lab},
            city={Grenoble},
            postcode={38000},
            country={France}}

\affiliation[3]{organization={Univ. Grenoble Alpes, Maths à Modeler}, 
            city={Grenoble},
            postcode={38000},
            country={France}}

\begin{abstract}
\noindent Let $G=(V,E)$ be a connected graph, and let $T \subseteq V$ be a subset of vertices. An orientation of $G$ is called $T$-odd if any vertex $v \in V$ has odd in-degree if and only if it is in $T$. Finding a $T$-odd orientation of $G$ can be solved in polynomial time as shown by Chevalier, Jaeger, Payan and Xuong (1983). Since then, $T$-odd orientations have continued to attract interest, particularly in the context of global constraints on the orientation. For instance, Frank and Kir\'aly (2002)  investigated $k$-connected $T$-odd orientations and raised questions about acyclic $T$-odd orientations. This problem is now recognized as an Egres problem and is known as the ``Acyclic orientation with parity constraints'' problem. Szegedy (2005) proposed a randomized polynomial algorithm to address this problem. An easy consequence of his work provides a polynomial time algorithm for planar graphs whenever $|T| = |V|-1$. Nevertheless, it remains unknown whether it exists in general. In this paper we contribute to the understanding of the complexity of this problem by studying a more general one. We prove that finding a $T$-odd acyclic orientation on graphs having some directed edges is \NP-complete.
\end{abstract}
\end{frontmatter}

\begin{keyword}
parity orientation, acyclic orientation, planar graph
\end{keyword}

\section{Introduction} \label{sec:introduction}

\noindent The study of constrained graph orientations and parity related problems has received significant interest in recent years. A notable problem intersecting these two domains is the acyclic orientation with parity constraints problem, which has been the subject of various research efforts. The first constraint considered in this problem is a local one, that involves the parity of the in-degree of vertices in an orientation. For a graph $G = (V,E)$ and a subset $T \subseteq V$, an orientation of $G$ is called $T$-odd if any vertex $v \in V$ has odd in-degree if and only if it is in $T$. Chevalier, Jaeger, Payan and Xuong (1983) \cite{chevalierOddRootedOrientations1983} proved that a $T$-odd orientation of an undirected graph exists if and only if $|E(G)| + |T| \equiv 0$ and that finding one can be done in polynomial time.\newline

\noindent Further advancements were made by Frank and Kir\'aly (2002) \cite{frankGraphOrientationsEdgeconnection2002} when adding a global constraint on the orientation, the $k$-arc-connectivity ie. orientations such that any pair of vertices is connected by $k$ arc-disjoints paths. They gave a characterization of undirected multigraphs $G = (V, E)$ having a $k$-arc-connected $T$-odd orientation for every subset $T \subseteq V$ with $|E| + |T| \equiv_2 0$.
It also raised questions on other global constraints such as acyclicity which can be seen as an other kind of connectivity constraint. This leads to the current formulation of Problem \ref{prob:AOP}. 
\begin{problem}[\textsc{Acyclic Orientation with parity constraints}]\label{prob:AOP}\ \newline
	\textbf{Instance:} A pair $(G,T)$ with $G=(V,E)$ an undirected graph and $T \subset V$ a subset of its vertex set.\newline
	\textbf{Question:} Does $G$ have an acyclic $T$-odd orientation?\newline
\end{problem}
\noindent This existence problem is referenced on the Egres platform\footnote{\url{https://lemon.cs.elte.hu/egres/open/Acyclic_orientation_with_parity_constraints}} and is known as the ``Acyclic orientation with parity constraints'' problem. It is easy to check it is in \NP. Indeed, given an orientation, one can verify in polynomial time whether it satisfies the in-degree constraints for each vertex and whether the orientation is acyclic. Still, it is not known whether Problem \ref{prob:AOP} is in \coNP. Note that no acyclic orientation exists if $G$ has loops. And if $G$ has multiple edges, they must be oriented the same way to preserve acyclicity and we could replace equivalently any odd multi-edge by a single edge and remove even multi-edges. Thus, we will focus on simple graphs.\newline


\noindent In 2005, Szegedy \cite{szegedyApplicationsWeightedCombinatorial2005} gave a randomized polynomial algorithm to decide whether a graph has an acyclic $T$-odd orientation for a given $T$. A key step is reducing to the case $|V\setminus T| = 1$. 


\begin{proposition}[2.11.5 in \cite{szegedyApplicationsWeightedCombinatorial2005}]\label{lem:reduceToApex}
	Let $(G, T)$ be an instance of Problem \ref{prob:AOP} \newline
	$G$ has an acyclic $T$-odd orientation if and only if $G'=(V\cup\{v\}, E\cup\{(v,u), u\in V\setminus T\})$ has an acyclic $V\setminus \{v\}$-odd orientation for some vertex $v \in V$.
\end{proposition}


\noindent Hence, the study of acyclic $T$-odd orientations when $|V\setminus T| = 1$ is of a lot of interest, and following this work, Kir\'aly and Kisfaludi-Bak gave in 2012 \cite{kiralyDualCriticalGraphsNotes} an algorithm deciding in polynomial time if a graph admits an  acyclic $V\setminus \{v\}$-odd orientation for some vertex $v \in V$ when the graph is either planar or $3$-regular (see corollary 1.41 and 2.17 in \cite{kiralyDualCriticalGraphsNotes}). However, in a more general setting, when $|V\setminus T| > 1$, the construction of $G'$ in Proposition \ref{lem:reduceToApex} may not preserve properties like planarity (except if $V\setminus T$ is included in a face)  or $3$-regularity. So Problem \ref{prob:AOP} remains open even on planar or $3$-regular graphs.\newline

\noindent A basic greedy algorithmic approach to Problem \ref{prob:AOP} would be to build a $T$-odd orientation iteratively by directing edges one by one. Thus, a step would consist of choosing an edge to direct in a partially directed graph. This leads to the study of the following problem:

\begin{problem}[\textsc{Acyclic Orientation with parity constraints on partially directed graphs}]\label{prob:PDAOP}\ \newline
	\textbf{Instance:} A pair $(G,T)$ where $G$ is a graph having a subset $A$ of edges with a fixed orientation and $T \subset V$ a subset of its vertex set.\newline
	\textbf{Question:} Does $G$ have an acyclic $T$-odd orientation preserving the fixed orientation of the pre-directed edges in $A$?
\end{problem}

\noindent A polynomial time algorithm to solve this problem on trees and on graphs of maximum degree $2$ can be done as follows. For a tree $G = (V,E)$ and $T\subset V$, there exists a unique $T$-odd orientation if and only if $|E(G)| + |T| \equiv 0$ and it is indeed acyclic. It can be obtained in polynomial time by \cite{chevalierOddRootedOrientations1983} and it only remains to check if the pre-directed edges in $A$ match this unique orientation. Now, a graph of maximum degree $2$ is a disjunction of paths and cycles. Dealing with the cycles is very similar to trees, except that for a given set $T$ verifying $|E(G)| + |T| \equiv 0$, a cycle admits exactly two $T$-odd orientations, where one can be obtained from the other by flipping all the arcs. It only remains to check the acyclicity and consistency with the pre-directed edges.\newline
\noindent Our next result shows that classes of graphs on which Problem \ref{prob:PDAOP} can be solved in polynomial time cannot be much more complex:

\begin{theorem}\label{thm:main}
    Problem \ref{prob:PDAOP} is \NP-complete even when restricted to planar graphs of maximum degree $3$ and such that any vertex not in $T$ has degree $2$.
\end{theorem}

\noindent As a consequence of this theorem, a basic greedy algorithmic approach is unlikely to exist for this problem, unless $\P = \NP$. On an other hand, Problem \ref{prob:AOP} is still open even for graphs being both planar and $3$-regular, and this result emphasizes once more the interest of those classes for Problem \ref{prob:AOP} and could lead to further results.

\section{Basic definitions} \label{sec:preliminaries}
\noindent A \textbf{partially directed graph} $G$ consists of a non-empty set $V(G)$ of vertices, a set $E(G)$ of edges (unordered pairs of vertices), and a set $A(G)$ of arcs (ordered pairs of vertices).  An edge will be denoted independently $uv$ or $vu$ with $u, v\in V$. Whereas, an arc will be denoted $\dir{uv}$ (equivalently $\ldir{vu}$) or $\dir{vu}$ (equivalently $\ldir{uv}$) according to the ordering of the vertices $u$ and $v$. A graph $G$ is said to be undirected if $A(G)$ is empty, and directed if $E(G)$ is empty.\newline

\noindent Given a graph $G$ and a set $S\subset E(G) \cup A(G)$, an \textbf{orientation} $O$ of $S$ is a set of arcs such that $S\cap A(G)\subseteq O$ and for any edge $uv\in S \cap E(G)$ we have that either $\dir{uv} \in O$ or $\dir{vu} \in O$. Given a subset $S'\subseteq S$ and an orientation $O$ of $S$, we denote $O(S')$ the orientation of $S'$ defined by $O(S')\subseteq O$. An {\bf orientation $O$ of a partially directed graph $G$} is an orientation of $E(G)\cup A(G)$ and it will be abusively denoted by $O(G)$ instead of $O(E(G)\cup A(G))$.

\noindent A set of arcs $A$ is said to be \textbf{acyclic} if the subgraph induced by $A$ does not contain any directed cycle. Similarly, a partially directed graph $G$ is acyclic if its arc set $A(G)$ is acyclic.

\noindent Given $X \subset V$, we define $\delta^\circ(X)$ to be the set of edges having exactly one vertex in $X$ and $\delta^+(X)$, $\delta^-(X)$ respectively to be the sets of arcs directed outwards and inwards from $X$. Then we define the \textbf{boundary} of $X$ to be $\delta(X) \overset{def}{=} \delta^\circ(X) \cup \delta^+(X) \cup \delta^-(X)$. We say that $\delta(X)$ is {\bf uniform} if it does not contain any edge and its arcs are directed either all outwards or all inwards from $X$.  For $X,Y\subset V(G)$, we define also the set $\delta(X,Y) \overset{def}{=} \delta(X)\cap \delta(Y)$. For any vertex $v \in V(G)$, we define the \textbf{out-degree} $d^+(v) = |\delta^+(\{v\})|$, the \textbf{in-degree} $d^-(v) = |\delta^-(\{v\})|$ and the \textbf{edge-degree} $d^\circ(v) = |\delta^\circ(\{v\})|$. The \textbf{degree} $d(v)$ is $|\delta(\{v\})|$.\newline

\noindent We abusively extend those notations to subgraphs. Let $H$ be a subgraph of a graph $G$, the \textbf{boundary} of $H$ will be noted $\delta(H)$ instead of $\delta(V(H))$. Similarly, $\delta^\circ(H)$, $\delta^+(H)$, $\delta^-(H)$ denote $\delta^\circ(V(H))$, $\delta^+(V(H))$, $\delta^-(V(H))$ respectively and for two subgraphs $H_1,H_2$ of a graph $G$, $\delta(H_1,H_2)$ denotes $\delta(V(H_1),V(H_2))$. We will further need to consider subgraphs along with their boundary, thus, we define for any subgraph $H$ of $G$, the subgraph induced by the set of edges and arcs $\delta(H)\cup E(H)\cup A(H)$ noted as $\Gamma_H$.\newline

\noindent Let $I=(G,T)$ be an instance of Problem \ref{prob:AOP}. We note $G(I) = G$ and $T(I) = T$. For simplicity, we note also $V(I) = V(G), E(I) = E(G)$ and $A(I) = A(G)$.
For the figures, we will refer to vertices in $T(I)$ as \textbf{black} vertices and vertices in $V(I) \bs T(I)$ as \textbf{white} vertices. It will also prove useful in the proofs to work with a slight generalization of $T$-odd orientations: given a graph $G$, a subgraph $S$ and a subset $T \subset V(G)$, an orientation $O$ is $T$-\textbf{odd on }$S$ if vertices of $V(S)$ have odd in-degree in $O(G)$ if and only if they belong to $T$. Note that in this definition, no constraints is imposed on the in-degree of the vertices in $V(G)\backslash V(S)$.\newline


\noindent We will also use the definition of \textbf{planar embedding} and \textbf{planar graph} that can be found in the book of Mohar and Thomassen \cite{moharGraphsSurfaces2001}. Roughly speaking, a planar embedding is a drawing of a graph in the plane such that no edges cross each other. A graph that admits a planar embedding is called planar. \newline

\section{\NP-completeness on partially directed graphs} \label{sec:partiallyDirectedGraphs}

\noindent In this section, we establish the \NP-completeness of the problem of finding an acyclic $T$-odd orientation of a partially directed graph. To demonstrate this, we employ a reduction from the well-known \NP-complete problem \PlanarThreeSAT.

\begin{problem}[\ThreeSAT] \label{prob:3SAT}
	\ \\
	\textbf{Instance:} A pair $(\cX, \cC)$ where $\cX = \{x_1, x_2, \ldots, x_n\}$ is a set of boolean variables and $\cC = \{c^1, c^2, \ldots, c^m\}$ is a set of clauses. Each clause $c^j$ is a disjunction of exactly three distinct literals, where a literal is either a variable from $\cX$ or its negation. No clause contains both a variable and its negation (otherwise it is trivially satisfied). \newline
	\textbf{Question:} Is there an assignment of the variables in $\cX$ such that all clauses in $\cC$ are satisfied?
\end{problem}

\noindent For any \ThreeSAT\ instance, we can construct what is known as its incidence graph. Let $(\mathcal{X}, \mathcal{C})$ be a \ThreeSAT\ instance. Its \textbf{incidence graph} $\cB_{\cX, \cC}$ is a bipartite graph with vertices corresponding to the variables and clauses of $(\mathcal{X}, \mathcal{C})$. An edge exists between a variable vertex $x_i$ and a clause vertex $c^j$ if and only if $x_i$ appears in $c^j$. Thus, for any $j$, $d(c^j) = 3$, and for any $i$, $d(x_i)$ corresponds to the number of clauses containing the variable $x_i$.

\begin{problem}[\PlanarThreeSAT] \label{prob:planar3SAT}
	\ \\ 
	\textbf{Instance:} An instance $(\cX, \cC)$ of \ThreeSAT\ and a planar embedding of $\cB_{\cX, \cC}$.\newline
	\textbf{Question:} Is there an assignment of the variables in $\cX$ such that all clauses in $\cC$ are satisfied?
\end{problem}

\begin{theorem}[Lichtenstein 1982 \cite{lichtensteinPlanarFormulaeTheir1982}]\label{thm:P3SAT}
    \PlanarThreeSAT~is \NP–complete.
\end{theorem}

\noindent Given a planar graph and one of its planar embeddings, let $\sigma_{v}$ be a clockwise cyclic permutation of the neighbors of $v$ according to the planar embedding. Thus $\sigma_{v}(i)$ returns the i-est neighbor of $v$ and $\sigma_{v}^{-1}(u)$ returns the index of the neighbor $u$ in the cyclic order.
For example, in Figure \ref{fig:incidence_graph}, we could define $\sigma_{c^4} = (x_2, x_5, x_4)$ and $\sigma_{x_5} = (c^5, c^3, c^4, c^2)$.
\begin{figure}[ht]
\begin{center}
	\begin{minipage}{4cm}
		$\cX = \{x_1, x_2, x_3, x_4, x_5\}$\newline
		$\cC = \{c^1, c^2, c^3, c^4, c^5\}$\newline
		$c^1 = x_1 \lor x_2 \lor x_3$\newline
		$c^2 = \neg x_1 \lor x_5 \lor x_2$\newline
		$c^3 = x_2 \lor \neg x_4 \lor x_5$\newline
		$c^4 = \neg x_2 \lor x_5 \lor x_4$\newline
		$c^5 = \neg x_3 \lor \neg x_2 \lor x_5$
	\end{minipage}
	\begin{minipage}{7cm}
		\centering
		\includegraphics[scale=1]{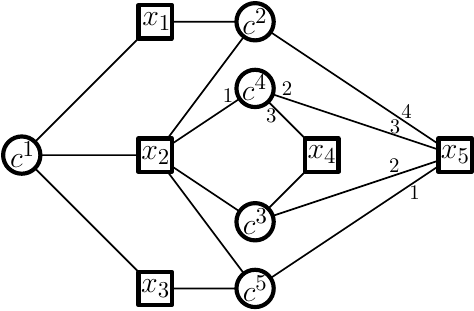}
	\end{minipage}
    \caption{Incidence graph of a \PlanarThreeSAT~instance. Here, $\sigma_{c^4} = (x_2, x_5, x_4)$ and $\sigma_{x_5} = (c^5, c^3, c^4, c^2)$}
    
	\label{fig:incidence_graph}
\end{center}
\end{figure}

\noindent Therefore, our objective is to reduce \PlanarThreeSAT~to our problem. To do so, we construct a graph from a \PlanarThreeSAT~instance such that finding an acyclic $T$-odd orientation of this graph is equivalent to finding a $True$ assignment of the \PlanarThreeSAT~instance.

\subsection{Graph construction} \label{sec:graphConstruction}
\noindent From  an instance $(\cX, \cC)$ of \PlanarThreeSAT, we construct an instance $I_{\cX, \cC}=(G(I_{\cX, \cC}), T(I_{\cX, \cC}))$ of Problem \ref{prob:PDAOP} based on a planar embedding of the incidence graph $\cB_{\cX, \cC}$ with a cyclic permutation $\sigma$ of the neighborhood of each vertex. We recall that we aim to build $I_{\cX, \cC}$ such that $G(I_{\cX, \cC})$ is a planar graph of maximum degree 3 and such that any vertex not in $T(I_{\cX, \cC})$ has degree 2. Essentially, $G(I_{\cX, \cC})$ has the same structure as $\cB_{\cX, \cC}$, where each variable vertex $x_i$ is replaced by a planar gadget $X_i$ with  $10.d(x_i)$ vertices and each clause vertex $c^j$ is replaced by a planar gadget $C^j$ with $12$ vertices. To guarantee that $G(I_{\cX, \cC})$ is also planar, we connect these gadgets as in the planar embedding of $\cB_{\cX, \cC}$ where each edge of $\cB_{\cX, \cC}$ corresponds to two parallel edges in $G(I_{\cX, \cC})$.

\begin{itemize}
	\item We construct the {\bf variable gadget} $X_i$ using a series of smaller gadgets.\par
	We use the term {\bf base gadget} to refer to the graph $M$ on $10$ vertices composed of a path on $6$ vertices with endpoints $u,\hat{u}$ plus a disjoint path on $4$ vertices with endpoints $s, t$. These two paths are linked with $4$ arcs as shown in Figure \ref{fig:M}. Finally, $T(M) = V(M)-\hat{u}$.

	To construct the variable gadget $X_i$, we combine $d(x_i)$ copies $M_i^{1}, \dots, M_i^{d(x_i)}$ of $M$ as follows. We add a superscript $k$ to the vertices of $M_i^{k}$ and a subscript $i$ to every vertex of $X_i$. Additionally, for $1 \leq k \leq d(x_i)$, we add an edge between the vertices $t_i^k$ and $s_i^{k+1}$ cyclically (i.e. with $s_i^{d(x_i)+1} = s_i^1$) as shown in Figure \ref{fig:X_i}. We define $T(X_i)$ to be the union of the $T(M_i^k)$. 

	\begin{figure}[ht]
	\centering
	\begin{subfigure}[b]{0.38\textwidth}
		\centering
		\ig[scale=0.7]{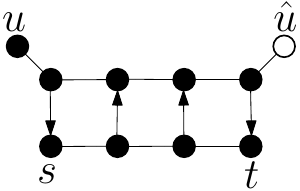}
	\caption{}
		\label{fig:M}
	\end{subfigure}
	\begin{subfigure}[b]{0.6\textwidth}
		\centering
		\ig[scale=0.7]{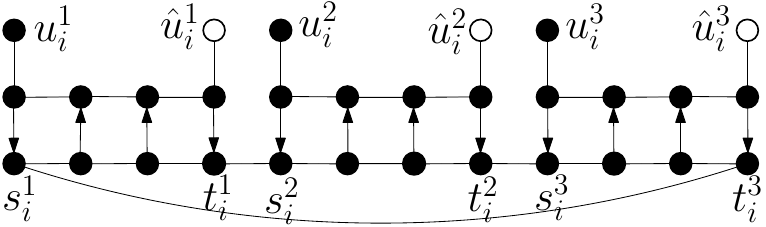}
		\caption{}
		\label{fig:X_i}
	\end{subfigure}
	\caption{Elements used to build variable gadgets: (a) Base gadget $M$; (b) Gadget $X_i$ with $d(x_i) = 3$.}
    \end{figure}

	\item For each clause vertex $c^j$ of $\cB_{\cX,\cC}$ with its neighboring vertices $\sigma_{c^j}(1)$, $\sigma_{c^j}(2)$, $\sigma_{c^j}(3)$ in $\cB_{\cX,\cC}$, we construct the {\bf clause gadget} $C^j$ as illustrated in Figure \ref{fig:clauseGadget}. The vertex set of $C^j$ is partitioned into $W_{C^j}\cup R_{C^j}$ where the graph induced by $W_{C^j}$ is the cycle on six vertices $w_{1}^j, \hat{w}_{1}^j, w_{2}^j, \hat{w}_{2}^j, w_{3}^j, \hat{w}_{3}^j$ with edges $w_{k+1}^j \hat{w}_{k+1}^j$ and $\hat{w}_{k+1}^j w_{k+2\bmod [3]}^j$ for all $k=0,1,2$, and the graph induced by $R_{C^j}$ is the stable set $v_{1}^j, \hat{v}_{1}^j, v_{2}^j, \hat{v}_{2}^j, v_{3}^j, \hat{v}_{3}^j$. Additionally, $C^j$ contains the matching $\delta(W_{C^j})$ defined by edges $v_{k}^j w_{k}^j$ and $\hat{v}_{k}^j \hat{w}_{k}^j$ for all $k=1,2, 3$. \par
	All vertices in $W_{C^j}$ belong to $T(C^j)$. For $k=1,2,3$, the vertices $v_{k}^j$ and $\hat{v}_{k}^j$ belong to $T(C^j)$ if and only if the $k$-th neighboring vertex of $c^j$ is expressed positively in the clause $c^j$.

	\begin{figure}[ht]
		\centering
		\includegraphics[scale=1]{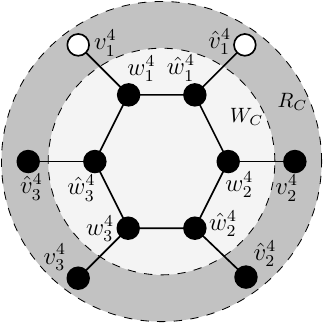}
		\caption{Gadget $C^4$ associated with the clause $c^4 = \neg x_2 \lor x_5 \lor x_4$ with $\sigma_{c^4} = (x_2, x_5, x_4)$}
		\label{fig:clauseGadget}
	\end{figure}

	\item For each edge $(x_i, c^j)$ in $\cB_{\cX,\cC}$, we add the two parallel edges $(u_i^{\sigma_{x_i}^{-1}(c^j)}, \hat{v}_{\sigma_{c^j}^{-1}(x_i)}^j)$ and $(\hat{u}_i^{\sigma_{x_i}^{-1}(c^j)}, v_{\sigma_{c^j}^{-1}(x_i)}^j)$ to the graph. Recall that $\sigma_{v}^{-1}(u)$ returns the index of a neighbor $u$ in the neighborhood of $v$ in $\cB_{\cX, \cC}$. This link between a clause and a variable gadget is illustrated in Figure \ref{fig:graphConstruction}.

	\item Finally we define $T(I_{\cX, \cC})$ to be the union of the $T(X_i)$ and $T(C^j)$ for all $i\in \{1,...,n\}$ and $j\in \{1,...,m\}$.
\end{itemize}

\begin{figure}[ht]
	\centering
	\begin{minipage}{4cm}
		$\cX = \{x_1, x_2, x_3, x_4, x_5\}$\newline
		$\cC = \{c^1, c^2, c^3, c^4, c^5\}$\newline
		$c^1 = x_1 \lor x_2 \lor x_3$\newline
		$c^2 = \neg x_1 \lor x_5 \lor x_2$\newline
		$c^3 = x_2 \lor \neg x_4 \lor x_5$\newline
		$c^4 = \neg x_2 \lor x_5 \lor x_4$\newline
		$c^5 = \neg x_3 \lor \neg x_2 \lor x_5$
	\end{minipage}
	\begin{minipage}{12cm}
		\centering
		\ig[scale=0.75]{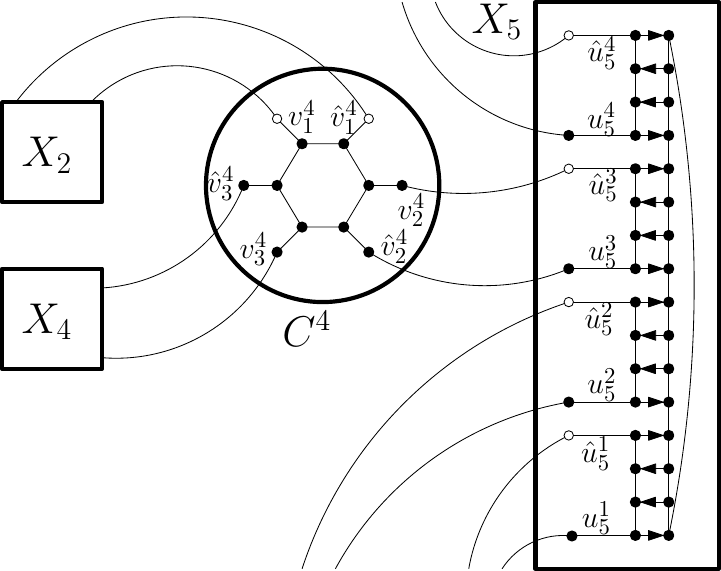}
	\end{minipage}

	\caption{\parbox[t]{0.60\textwidth}{\centering%
		Construction of part of $I_{\cX, \cC}$ from a \PlanarThreeSAT~instance $(\cX, \cC)$. Some of the gadgets are simplified.}
	}
	\label{fig:graphConstruction}

\end{figure}

\noindent Given $(\cX, \cC)$ a \PlanarThreeSAT~instance, we call $I_{\cX,\cC}$ its corresponding instance of Problem \ref{prob:PDAOP}. We will now prove that an acyclic $T(I_{\cX,\cC})$-odd orientation of $I_{\cX,\cC}$ exists if and only if there exists a $True$ assignment of $(\cX,\cC)$.

\subsection{Property of the variable gadget} \label{sec:variableGadget}
\noindent The purpose of the variable gadget $X$ is to ensure that all edges in $\delta(X)$ are directed identically in any acyclic orientation $T$-odd on $X$, i.e. $\delta(X)$ in $O(\Gamma_X)$, is uniform.\newline

\noindent Let us first prove some elementary properties of $M$:
\begin{lemma} \label{lem:M4}
	Let $(\cX, \cC)$ be a \PlanarThreeSAT~instance and $I_{\cX,\cC}=(G,T)$ be the corresponding instance of Problem \ref{prob:PDAOP}. Let $M$ be any base gadget that appears in $G$.

    \noindent $\Gamma_M$ admits exactly two acyclic orientations $T$-odd on $M$.\newline
\end{lemma}

\begin{proof}
    Notice that $V(\Gamma_M)$ contains all vertices of $V(M)$ plus four vertices of $V(G)\setminus V(M)$ adjacent to $u, \hat{u}, s$ and $t$. Let $O$ be an acyclic $T$-odd on $M$ orientation of $\Gamma_M$. Because each vertex has to satisfy the condition $T$-odd on $M$ and is adjacent to exactly two undirected edges, choosing the orientation of one of those edges will force the orientation of the other. An as example, because $a\in T$ and $\dir{as} \in O$, we have $\dir{ua} \in O \iff \dir{ab} \in O$. Repeating this reasoning on the other vertices we deduce:

    $$ \dir{ua} \in O \iff \dir{ab} \in O \iff \ldir{bc} \in O \iff \dir{cd} \in O \iff \dir{d\hat{u}} \in O$$
    $$ \dir{se} \in O \iff \dir{ef} \in O \iff \dir{ft} \in O$$
    
    \noindent Furthermore, the acyclicity condition on $O$ imposed on the sequences of vertices $(seba)$ and $(tfcd)$ gives:
    $$ \dir{se} \in O \Rightarrow \dir{ab} \in O$$
    $$ \ldir{ft} \in O \Rightarrow \dir{cd} \in O$$
    
    \noindent From those constraints, we get only the two possible acyclic $T$-odd on $M$ orientations of $\Gamma_M$ given in Figure \ref{fig:M_orientations}.
\end{proof}

\begin{figure}[ht]
	\centering
	\begin{minipage}{8cm}
		\centering
		\ig[scale=1, page=1]{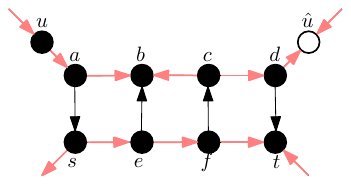}
	\end{minipage}
	\begin{minipage}{8cm}
		\centering
		\ig[scale=1, page=2]{./pics/M_orientations.pdf}
	\end{minipage}

	\caption{The two possible acyclic orientations of $\Gamma_M$}
	\label{fig:M_orientations}
\end{figure}

\noindent This leads to the following corollary:

\begin{corollary} \label{cor:variableGadget}
	Let $(\cX, \cC)$ be a \PlanarThreeSAT~instance and $I_{\cX, \cC} =(G,T)$ be the corresponding instance of Problem \ref{prob:PDAOP}. Let $X$ be any variable gadget that appears in $G$.\newline

	\noindent $\Gamma_X$ admits exactly two acyclic orientations $T$-odd on $X$.\newline
	\noindent Additionally, one satisfies $\delta(X) = \delta^+(X)$ and the other satisfies $\delta(X) = \delta^-(X)$.
\end{corollary}

\begin{proof}
    Similarly to the previous result, one can check that the only orientations of $\Gamma_X$ being both acyclic and $T$-odd on $X$ are the two depicted in Figure \ref{fig:X_orientation_out} and \ref{fig:X_orientation_in}. Choosing the orientation of any edge forces the orientation of each $M^j$. It also preserves acyclicity on the overall orientation as, for all $j$, either $s^j$ or $t^j$ has  $d^+ = 0$.
\end{proof}

\noindent We denote by $O(\Gamma_X)_{out}$  (resp. $O(\Gamma_X)_{in}$) the acyclic orientation of $\Gamma_X$ that is $T$-odd on $X$ and satisfy $\delta(X) = \delta^+(X)$ (resp. $\delta(X) = \delta^-(X)$).

\begin{figure}[ht]
	\centering
	\begin{minipage}{8cm}
		\centering
		\ig[scale=0.5, page=1]{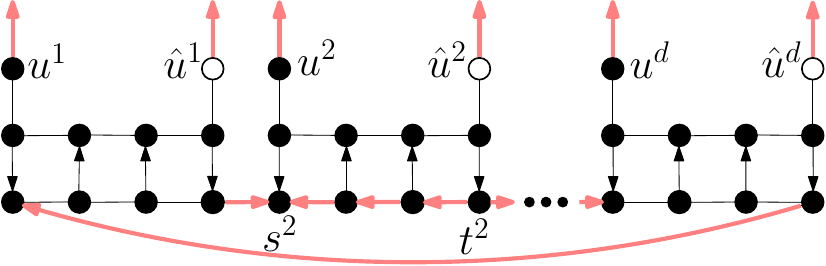}
		\caption{$O(\Gamma_X)_{out}$}
		\label{fig:X_orientation_out}

	\end{minipage}
	\begin{minipage}{8cm}
		\centering
		\ig[scale=0.5, page=2]{./pics/X_orientations.pdf}
		\caption{$O(\Gamma_X)_{in}$}
		\label{fig:X_orientation_in}

	\end{minipage}
\end{figure}

\FloatBarrier

\begin{figure}[H]
	\centering
	\begin{minipage}[t]{10.2cm}
	\vspace{0pt}
	\subsection{Property of the clause gadget} \label{sec:clauseGadget}

	In this section, we are going to look at clause gadgets $C$ only as they appear in the construction detailed earlier. Thus, to simplify the notations, the neighbors of $v_i, \hat{v}_i$ in $\Gamma_C\setminus V(C)$ are renamed $\gamma_{v_i}, \gamma_{\hat{v}_i}$ instead of, respectively $u_i^{\sigma_{c}^{-1}(x_i)}, \hat{u}_i^{\sigma_{c}^{-1}(x_i)}$, as shown in Figure \ref{fig:Noted_C} and we will abusively say that they belong to the variable gadget $X_i$ instead of $X_{\sigma_{c}^{-1}(x_i)}$. Recall also that $W_C=\{w_1, \hat{w}_1, w_2, \hat{w}_2, w_3, \hat{w}_3\}$ is the vertex set of $C$ inducing a cycle.
	\end{minipage}
	\begin{minipage}[t]{6cm}
    	\centering
		\vspace{-\baselineskip}
    	\includegraphics[scale=1]{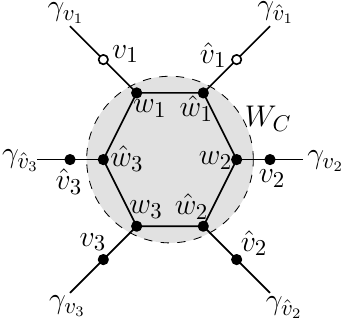}
    	\caption{Notations for $C$}
    	\label{fig:Noted_C}
	\end{minipage}
\end{figure}

\begin{lemma}\label{lem:clauseGadget}
    Let $(\cX, \cC)$ be a \PlanarThreeSAT~instance and $I_{\cX, \cC} = (G,T)$ be the corresponding instance of Problem \ref{prob:PDAOP}.
    Then, there exists a $T$-odd orientation $O$ of $I_{\cX, \cC}$ such that $O(\Gamma_{X_i})$ is acyclic for all $i \in \{1, \dots, n\}$.\newline
    
    \noindent Furthermore, $O$ is an acyclic orientation if and only if it satisfies $\delta^-(W_C) \neq \emptyset$ for all clause gadgets $C$.
\end{lemma}

\begin{proof}
From Corollary \ref{cor:variableGadget}, we know that each variable gadget $\Gamma_{X_i}$ has an acyclic orientation $T$-odd on $X_i$. Let $O_i$ be such an orientation for all $i \in \{1, \dots, n\}$. We will show that there exists an orientation $O$ of $I_{\cX, \cC}$ such that $O_i \subset O$ for all $i$ and $O(\Gamma_{C^j})$ is $T$-odd on $C^j$ for all $j \in \{1, \dots, m\}$.\newline

\noindent Let $C^j$ be any clause gadget. We construct an orientation $O^j$ of $\Gamma_{C^j}$ that is $T$-odd on $C^j$ and such that $O_i(\delta(X_i,C^j)) = O^j(\delta(X_i,C^j))$.

\noindent First, observe that $\gamma_{v_i}$ and $\gamma_{\hat{v}_i}$ belong to the same variable gadget $X_i$. By Corollary \ref{cor:variableGadget}, either both $\dir{\gamma_{v_i} v_i}$ and $\dir{\gamma_{\hat{v}_i} \hat{v}_i}$ belong to $O_i$, or both $\overleftarrow{\gamma_{v_i} v_i}$ and $\overleftarrow{\gamma_{\hat{v}_i} \hat{v}_i}$ belong to $O_i$.Then, we fix $\dir{\gamma_{v_i} v_i}$ and $\dir{\gamma_{\hat{v}_i} \hat{v}_i}$ in $O^j$ if and only if they belong to $O_i$.

\noindent Since $v_i$ and $\hat{v}_i$ either both belong to $T$ or both do not, we ensure the in-degree condition is respected in $O^j$ by directing $v_i w_i$ and $\hat{v}_i \hat{w}_i$ as follows:
\begin{itemize}
    \item If $\dir{\gamma_{v_i} v_i}, \dir{\gamma_{\hat{v}_i} \hat{v}_i} \in O^j$ and $v_i, \hat{v}_i \in T$, then include $\dir{v_i w_i}$ and $\dir{\hat{v}_i \hat{w}_i}$ in $O^j$.
    \item If $\overleftarrow{\gamma_{v_i} v_i}, \overleftarrow{\gamma_{\hat{v}_i} \hat{v}_i} \in O^j$ and $v_i, \hat{v}_i \notin T$, then include $\dir{v_i w_i}$ and $\dir{\hat{v}_i \hat{w}_i}$ in $O^j$.
\end{itemize}
Thus, up to rotation, $\delta(W_{C^j})$ can be directed in exactly four ways in $O^j$ (illustrated with red arcs, see $a_0,a_1,a_2,a_3$ in Figure \ref{fig:C_orientation}). In each of those cases, there are exactly two possible ways to direct the remaining edges of $E(C^j)$ to obtain an orientation $T$-odd on $C^j$ (illustrated with blue arcs in Figure \ref{fig:C_orientation}). Let $O^j$ be any of those two orientations.\newline
Now, let $O = \bigcup_{i=1}^{n} O_i \cup \bigcup_{j=1}^{m} O^j$. We have that $O$ is an orientation of $I_{\cX, \cC}$ since $O_i(\delta(X_i,C^j)) = O^j(\delta(X_i,C^j))$ for all $i \in \{1, \dots, n\}$ and $j \in \{1, \dots, m\}$. And $O$ is $T$-odd since $O_i$ and $O^j$ are $T$-odd on $X_i$ and $T$-odd on $C^j$ respectively, for all $i \in \{1, \dots, n\}$ and $j \in \{1, \dots, m\}$.\newline

\noindent Also, notice that in $I_{\cX, \cC}$, there is no edge between two clause gadgets. Each variable gadget is directed in an acyclic $T$-odd way in $O$ and their boundary is uniform by Corollary \ref{cor:variableGadget}. Thus, if $O$ ever contains a directed cycle it must be included in some clause gadget $C^j$. But one can notice in Figure \ref{fig:C_orientation} that $O(C^j)$ contains a cycle if and only if $\delta^-(W_{C^j}) = \emptyset$ (see case $a_0$). Therefore, $O$ is acyclic if and only if $\delta^-(W_{C^j}) \neq \emptyset$ for all clause gadgets $C^j$ with $j \in \{1, \dots, m\}$.
\begin{figure}[ht]
\centering
\includegraphics[scale=1.1]{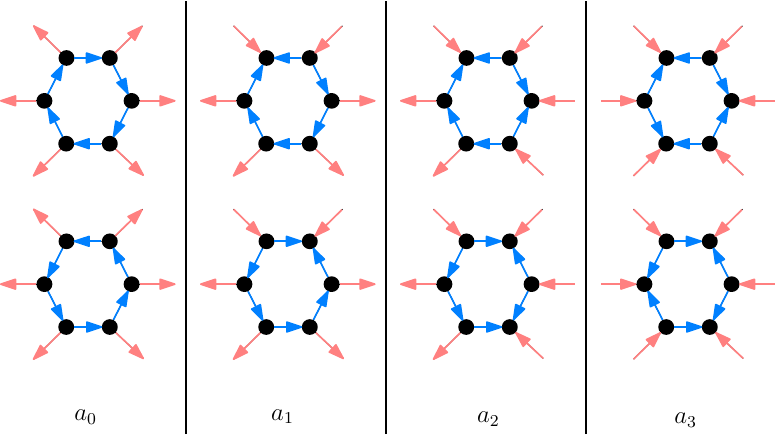}
\caption{Possible $T(I_{\cX, \cC})$-odd orientation of $E(C)$}
\label{fig:C_orientation}
\end{figure}

\end{proof}

\subsection{Proof of the reduction} \label{sec:correction}

\noindent In this section, we establish the following result expressing the reduction from \PlanarThreeSAT~to Problem \ref{prob:PDAOP}.\newline

\noindent Recall that the acyclic $T$-odd orientation problem is in \NP. Given an orientation, one can verify in polynomial time whether it satisfies the parity constraints for each vertex and whether the orientation is acyclic.

\begin{theorem}\label{thm:equiv}
	Let $(\cX, \cC)$ be a \PlanarThreeSAT~instance and $I_{\cX, \cC} = (G,T)$ be the corresponding instance of Problem \ref{prob:PDAOP}.\newline
	$I_{\cX, \cC}$ has an acyclic $T$-odd orientation if and only if  $(\cX, \cC)$ has a $True$ assignment.
\end{theorem}

\begin{proof}
\noindent Let $(\cX, \cC)$ be a \PlanarThreeSAT~instance and $I_{\cX, \cC} = (G,T)$ be the corresponding instance of Problem \ref{prob:PDAOP}.\newline

\noindent Let $\alpha$ be an assignment of $(\cX, \cC)$ and $O$ be a $T$-odd orientation of $I_{\cX, \cC}$ such that $O(\Gamma_{X_i})$ is acyclic for all $i\in \{1,...,n\}$. We say that $\alpha$ (or respectively $O$) is {\bf compatible} to $O$ (respectively $\alpha$) when, for all $i\in \{1,...,n\}$, we have that:
$$\alpha(x_i)= True \ \hbox{if and only if} \ O(\Gamma_{X_i})=O(\Gamma_{X_i})_{out}.$$

\noindent  Observe that, by Corollary \ref{cor:variableGadget}, if $\alpha$ and $O$ are compatible, then we also have, for all $i=1,...,n$, that: $$\alpha(x_i)= False \ \hbox{if and only if} \ O(\Gamma_{X_i})=O(\Gamma_{X_i})_{in}.$$

\noindent Notice also, from Corollary \ref{cor:variableGadget}, that any $T$-odd orientation $O$ of $I_{\cX, \cC}$ such that $O(\Gamma_{X_i})$ is acyclic for all $i\in \{1,...,n\}$ has a compatible assignment of $(\cX, \cC)$, and from Lemma \ref{lem:clauseGadget}, that any assignment of $(\cX, \cC)$ has a compatible orientation of $I_{\cX, \cC}$.\newline

\noindent Moreover, by Lemma \ref{lem:clauseGadget}, a $T$-odd orientation $O$ of $I_{\cX, \cC}$ such that $O(\Gamma_{X_i})$ is acyclic for all $i\in \{1,...,n\}$ is acyclic if and only if $O(\Gamma_{C^j})$ satisfies $\delta^-(W_{C^j}) \neq \emptyset$ for all $j \in \{1, \dots, m\}$.

\noindent And, since a $True$ assignment of $(\cX, \cC)$ is an assignment which satisfies each clause in $\cC$, to prove Theorem \ref{thm:equiv}, it is enough to check that for any $\alpha$ and $O$ compatible and any clause $c^j$, we have that: $$\alpha \ \hbox{satisfies} \ c^j \ \hbox{if and only if} \ \delta^-(W_{C^j}) \neq \emptyset \ \hbox{in} \ O(I_{\cX, \cC}).$$

\noindent Let $c^j\in \cC$. Without loss of generality, one may assume that $c^j= l_1\vee l_2\vee l_3$ with either $l_i=x_i$ or $l_i=\neg x_i$, for $i=1, ...,3$. By definition of the disjunction $\vee$, we have:

$$c^j \text{ is  satisfied by }\alpha \iff \exists l_i\in c^j, \begin{cases}
l_i=x_i \text{ and } \alpha(x_i)=True, \text{ or }\\
l_i=\neg x_i \text{ and } \alpha(x_i)=False.
\end{cases}
$$

\noindent By construction of $I_{\cX, \cC}$ and since $O$ is compatible to $\alpha$, this is equivalent to:

$$ \exists i, \begin{cases}
v^j_i \in T(I_{\cX, \cC}) \text{ and } O(\hat{u}_i^{\sigma^{-1}_{x_i}(c^j)}v^j_i)=\dir{\hat{u}_i^{\sigma^{-1}_{x_i}(c^j)}v^j_i}  \text{ or }\\
v^j_i \in V(I_{\cX, \cC})\setminus T(I_{\cX, \cC}) \text{ and } O(\hat{u}_i^{\sigma^{-1}_{x_i}(c^j)}v^j_i)=\ldir{\hat{u}_i^{\sigma^{-1}_{x_i}(c^j)}v^j_i}.
\end{cases}
$$

\noindent Since $v^j_i$ is of degree $2$, by the parity condition on the in-degree of $v^j_i$, this is equivalent to:

$$\exists i, O(v^j_iw_i^j)=\dir{v_i^jw_i^j}.$$

\noindent This is precisely $\delta^-(W_{C^j})\neq \emptyset$ in $O(I_{\cX, \cC})$. \newline
Thus, for $(O,\alpha)$ a pair of compatible orientation and assignment, $O$ is an acyclic $T$-odd orientation of $I_{\cX, \cC}$ if and only if $\alpha$ is a $True$ assignment of $(\cX, \cC)$.
\
\end{proof}

\noindent One can check that the construction of $I_{(\cX, \cC)}$ from $(\cX, \cC)$ can be done in polynomial time. Thus, our main result, Theorem \ref{thm:main}, is a direct consequence of Theorems \ref{thm:equiv} and \ref{thm:P3SAT}.

\section{Conclusion}
\noindent We proved that finding an acyclic $T$-odd orientation in a partially directed graph is \NP-complete, even when $G$ is planar of maximum degree 3, and vertices in $V(G)-T$ have degree 2.  However a polynomial time algorithm exists for trees or graphs of maximum degree $2$. So classes of graphs on which Problem \ref{prob:PDAOP} can be solved in polynomial time cannot be much more complex.\newline
We could also discuss the complexity of Problem \ref{prob:PDAOP} parameterized by the size of $T$. Unfortunately, the problem remains $\NP$-complete even when $T = \emptyset$. Indeed, from $(\cX, \cC)$ a \PlanarThreeSAT~instance and $I_{\cX, \cC} = (G,T)$ its corresponding instance of Problem \ref{prob:PDAOP}, any vertex $v\in T$ of degree $2$ could be contracted into a single edge resulting in an equivalent instance $I'_{\cX, \cC} = (G',T')$ where each $v\in V(G')$ has degree $3$ if $v\in T'$ and degree $2$ otherwise. Flipping the orientation of all edges in an acyclic $T'$-odd orientation of $I'_{\cX, \cC}$ leads to an acyclic $\emptyset$-odd orientation of $I'_{\cX, \cC}$. Thus, even finding acyclic $\emptyset$-odd orientation on partially directed planar graphs of maximum degree $3$ is \NP-complete. However, in our reduction from \PlanarThreeSAT~to Problem \ref{prob:PDAOP}, the constructed graph contains a linear number of arcs, which raises questions about the complexity of Problem \ref{prob:PDAOP} when parameterized by the number of arcs.\newline

\noindent Going back to Problem \ref{prob:AOP}, there exists a deterministic polynomial time algorithm to determine whether $G$ admits an acyclic $T$-odd orientation whenever $G$ is either planar or $3$-regular and $|T| = |V(G)|-1$ \cite{kiralyDualCriticalGraphsNotes}. But for a general $T \subset V$ the problem remains open, even if a randomized polynomial time algorithm is known \cite{szegedySymplecticSpacesEarDecomposition2006}.
Our result shows that a basic greedy algorithmic approach is not likely to exists. Overall, determining the complexity of Problem \ref{prob:AOP} on those classes is challenging as it belongs to the frontier of $\P$, $\RP$ and $\NP$ classes of problems (see Figure \ref{fig:complexity_classes}). 

\begin{figure}[ht]
	\centering
	\ig[scale=0.7, page=1]{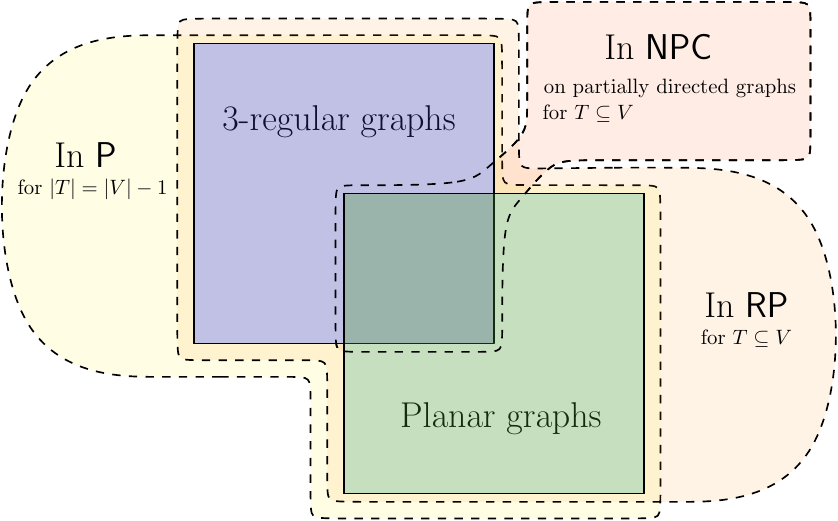}

	\caption{Complexity of Problem \ref{prob:AOP} on $3$-regular and planar graphs}
	\label{fig:complexity_classes}
\end{figure}

\bibliographystyle{elsarticle-num}
\bibliography{biblio}

@article{chevalierOddRootedOrientations1983,
title = {Odd Rooted Orientations and Upper-Embeddable Graphs},
journal = {North-Holland Mathematics Studies},
publisher = {North-Holland},
volume = {75},
pages = {177-181},
year = {1983},
booktitle = {Combinatorial Mathematics},
doi = {10.1016/S0304-0208(08)73385-1},
author = {O. Chevalier and F. Jaeger and C. Payan and N.H. Xuong}
}

@phdthesis{szegedyApplicationsWeightedCombinatorial2005,
author = {{C. Szegedy}},
title = {Some Applications of the Weighted Combinatorial Laplacian},
school = {Rheinische Friedrich-Wilhelms-Universität Bonn},
year = 2005,
url = {https://hdl.handle.net/20.500.11811/2260}
}

@article{frankGraphOrientationsEdgeconnection2002,
  title = {Graph {{Orientations}} with {{Edge-connection}} and {{Parity Constraints}}},
  author = {Frank, Andr{\'a}s and Kir{\'a}ly, Zolt{\'a}n},
  year = {2002},
  journal = {Combinatorica},
  volume = {22},
  number = {1},
  pages = {47--70},
  issn = {0209-9683, 1439-6912},
  doi = {10.1007/s004930200003},
  copyright = {http://www.springer.com/tdm}
}

@article{lichtensteinPlanarFormulaeTheir1982,
  title = {Planar {{Formulae}} and {{Their Uses}}},
  author = {Lichtenstein, David},
  year = {1982},
  journal = {SIAM Journal on Computing},
  volume = {11},
  number = {2},
  pages = {329--343},
  publisher = {{Society for Industrial and Applied Mathematics}},
  issn = {0097-5397},
  doi = {10.1137/0211025},
}

@article{szegedySymplecticSpacesEarDecomposition2006,
  title = {Symplectic {{Spaces And Ear-Decomposition Of Matroids}}},
  author = {Szegedy, Bal{\'a}zs and Szegedy, Christian},
  year = {2006},
  month = jun,
  journal = {Combinatorica},
  volume = {26},
  number = {3},
  pages = {353--377},
  issn = {1439-6912},
  doi = {10.1007/s00493-006-0020-3},
  langid = {english},
}

@article{kiralyDualCriticalGraphsNotes,
  title = {Dual-{{Critical Graphs}} -- {{Notes}} on Parity Constrained Acyclic Orientations},
  author = {Kiraly, Zoltan and {Kisfaludi-Bak}, Sandor},
  institution = {Egervary Research Group},
  year = {2012},
  url = {https://egres.elte.hu/tr/egres-12-07}
}

@book{moharGraphsSurfaces2001,
  title = {Graphs on Surfaces},
  author = {Mohar, Bojan and Thomassen, Carsten},
  year = {2001},
  month = aug,
  isbn = {978-0-8018-6689-0},
  doi = {10.56021/9780801866890}
}

\end{document}